\documentclass[11pt]{article}

\usepackage[margin=1in]{geometry}
\usepackage{microtype}

\usepackage[numbers,sort&compress]{natbib}

\usepackage[colorlinks=true, citecolor=blue, linkcolor=blue, urlcolor=blue]{hyperref}

\usepackage{algorithm}
\usepackage[noend]{algpseudocode}
\algrenewcommand\algorithmicrequire{\textbf{Inputs:}}
\algrenewcommand\algorithmicensure{\textbf{Outputs:}}
\makeatletter
\renewcommand{\ALG@beginalgorithmic}{\small}
\makeatother
\usepackage{amsmath,amssymb,amsthm,amsfonts}
\usepackage{graphicx}
\usepackage{booktabs}
\usepackage{subcaption}
\usepackage{xcolor}
\usepackage{float}

\newtheorem{theorem}{Theorem}

\newtheorem{proposition}{Proposition}

\usepackage{graphicx}
\usepackage{caption}
\usepackage{authblk}
\title{Neural Tangent Kernels for Complex Genetic Risk Prediction:\\
Bridging Deep Learning and Kernel Methods in Genomics}

\author{Heng Ge \quad Qing Lu}
\affil{Department of Biostatistics, University of Florida \\
\texttt{heng.ge@ufl.edu, lucienq@ufl.edu}}
\begin{document}
\date{}
\maketitle

\begin{abstract}
 Given the complexity of genetic risk prediction, there is a critical need for the development of novel methodologies that can effectively capture intricate genotype–phenotype relationships (e.g., nonlinear) while remaining statistically interpretable and computationally tractable. We develop a Neural Tangent Kernel (NTK) framework to integrate kernel methods into deep neural networks for genetic risk prediction analysis. We consider two approaches: \emph{NTK-LMM}, which embeds the empirical NTK in a linear mixed model with variance components estimated via minimum quadratic unbiased estimator (MINQUE), and \emph{NTK-KRR}, which performs kernel ridge regression with cross-validated regularization. Through simulation studies, we show that NTK-based models outperform the traditional neural network models and linear mixed models. By applying NTK to endophenotypes (e.g., hippocampal volume) and AD-related genes (e.g., APOE) from Alzheimer's Disease Neuroimaging Initiative (ADNI), we found that NTK achieved higher accuracy than existing methods for hippocampal volume and entorhinal cortex thickness. In addition to its accuracy performance, NTK has favorable optimization properties (i.e., having a closed-form or convex training) and generates interpretable results due to its connection to variance components and heritability. Overall, our results indicate that by integrating the strengths of both deep neural networks and kernel methods, NTK offers competitive performance for genetic risk prediction analysis while having the advantages of interpretability and computational efficiency.
\end{abstract}

\section{Introduction}
Accurate genetic risk prediction is central to precision medicine \citep{collins2015new,ashley2015precision}, yet for most complex diseases currently discovered loci explain only a modest fraction of heritability \citep{chatterjee2013projecting}. Because complex traits likely arise from many variants acting through intertwined biological pathways, modeling sets of variants jointly—potentially with interactions—can improve prediction \citep{neale2004future,chatterjee2006powerful,wu2011rare}. In practice, however, genotype data present severe statistical and computational challenges: the parameter space scales rapidly, genotype–phenotype relationships are often non-linear, the noise-to-signal ratio is high, and computational costs can be prohibitive \citep{kang2010variance,lippert2011fast,speed2014multiblup,weissbrod2016multikernel,wen2020multikernel}.

Linear mixed models (LMMs) have become a workhorse for genetic prediction at scale. Rather than estimating per-variant fixed effects, LMMs encode the assumption that genetic similarity implies phenotypic similarity via a kernel (or genetic relationship) matrix, yielding genomic BLUP (gBLUP) and its extensions \citep{yang2010common}. To increase flexibility, multi-component and multi-kernel variants partition the genome and assign separate variance components or kernels to different groups, thereby capturing heterogeneity across annotations or non-linear effects through pre-specified kernels \citep{speed2014multiblup,weissbrod2016multikernel,wen2020multikernel}. While effective, these models can be computationally demanding when many components are included. Moment-based methods provide computationally efficient alternatives to traditional likelihood-based approaches. While restricted maximum likelihood (REML) and maximum likelihood estimation (MLE) remain standard, they become computationally prohibitive at scale. Minimum Norm Quadratic Unbiased Estimation (MINQUE) and related General Method of Moments (GMM) approaches offer faster computation without sacrificing accuracy, making them particularly suitable for large-scale genetic studies \citep{rao1970estimation,rao1971estimation,rao1972estimation,shen2021kernel,wang2022penalized}. Nevertheless, standard LMMs remain limited when non-linear and hierarchical effects are prominent, and multi-kernel designs trade statistical flexibility against computational burden.

Deep neural networks provide powerful non-linear function classes \citep{Goodfellow-et-al-2016} but can be difficult to train reliably in high dimensions and are less directly connected to heritability-oriented inference \citep{van2020bayesian,singh2023fsnet}. The Neural Tangent Kernel (NTK) bridges these paradigms: in the infinite-width (lazy training) limit, gradient descent on a neural network corresponds to kernel regression under a data-dependent kernel determined by the architecture \citep{jacot2018neural}. This connection preserves the architectural inductive structure of deep networks while recovering convex optimization and closed-form predictors in the kernel regime. For genetic prediction, NTK thus offers a principled way to import expressive, hierarchical representations into statistically grounded kernel and mixed-model frameworks—potentially improving accuracy on traits with non-additive or saturating effects, without sacrificing interpretability or stability.

Building on the above motivation, we make three contributions: (i) we introduce two NTK-based predictors tailored to genomics: \emph{NTK-LMM}, which embeds the empirical NTK into an LMM with variance components estimated via MINQUE \citep{rao1970estimation,rao1971estimation,rao1972estimation}, and \emph{NTK-KRR}, which applies kernel ridge regression \citep{shawe2004kernel} with cross-validated regularization; (ii) we test our methods using simulations based on real genetic data patterns. Our testing framework examines, under different genetic dimension, how well the methods handle different types of relationships between genes and traits—from simple linear effects to complex nonlinear scenarios. This comprehensive testing ensures the methods work across diverse biological situations; and (iii) we apply our approach to brain imaging data from the Alzheimer's Disease Neuroimaging Initiative \citep{saykin2010alzheimer}, focusing on a well-studied genetic region (APOE–TOMM40–APOC1). This real-world test shows when our NTK method outperforms traditional approaches and when all methods perform similarly. Together, these results demonstrate that NTK offers a practical way to harness the power of deep learning for genetic prediction while keeping the interpretability and reliability that researchers need. An anonymized implementation of our method is provided in the supplementary material.

\section{Methodologies}
\subsection{Neural Tangent Kernel}

The Neural Tangent Kernel provides a theoretical framework connecting the behavior of infinitely wide neural networks to kernel methods \citep{jacot2018neural}. Consider a fully connected neural network $f(\mathbf{x}; \boldsymbol{\theta})$ with parameters $\boldsymbol{\theta}$ initialized from a Gaussian distribution, where $\mathbf{x} \in \mathbb{R}^d$ represents the input vector (e.g., genetic SNPs). As the width of hidden layers approaches infinity, the network's training dynamics under gradient descent can be characterized by a fixed kernel function.

The NTK is defined as:
\begin{equation}
\Theta(\mathbf{x}, \mathbf{x}') = \left\langle \nabla_{\boldsymbol{\theta}} f(\mathbf{x}; \boldsymbol{\theta}), \nabla_{\boldsymbol{\theta}} f(\mathbf{x}'; \boldsymbol{\theta}) \right\rangle
\end{equation}
where the inner product is taken over all network parameters. This kernel quantifies the similarity between two inputs based on how similarly they influence the network's parameters during training.

Like classical kernels (Gaussian, polynomial), the NTK defines a similarity measure between data points, but one that emerges from the neural network's architecture rather than being explicitly designed. Just as a Gaussian kernel $k(\mathbf{x}, \mathbf{x}') = \exp(-\|\mathbf{x} - \mathbf{x}'\|^2/2\sigma^2)$ measures similarity through Euclidean distance, the NTK measures similarity through how similarly the network's parameters affect predictions at different input points. Crucially, the NTK belongs to the class of reproducing kernel Hilbert space (RKHS) kernels, satisfying positive definiteness and enabling the representer theorem \citep{scholkopf2001generalized}. This means any function learned by an infinitely wide neural network through gradient descent can be expressed as:
\begin{equation}
f(\mathbf{x}) = \sum_{i=1}^{n} \alpha_i \Theta(\mathbf{x}, \mathbf{x}_i)
\end{equation}
where $\{\mathbf{x}_i\}_{i=1}^n$ are training points and $\alpha_i$ are learned coefficients.

For a network with $L$ layers and width $m$, the NTK can be computed recursively. Let $\Sigma^{(0)}(\mathbf{x}, \mathbf{x}') = \mathbf{x}^T\mathbf{x}'/d$ where $d$ is the input dimension. For layers $\ell = 1, \ldots, L$:
\begin{equation}
\begin{aligned}
\Lambda^{(\ell)}(\mathbf{x}, \mathbf{x}') &= \begin{pmatrix} 
\Sigma^{(\ell-1)}(\mathbf{x}, \mathbf{x}) & \Sigma^{(\ell-1)}(\mathbf{x}, \mathbf{x}') \\ 
\Sigma^{(\ell-1)}(\mathbf{x}', \mathbf{x}) & \Sigma^{(\ell-1)}(\mathbf{x}', \mathbf{x}') 
\end{pmatrix} \\
\Sigma^{(\ell)}(\mathbf{x}, \mathbf{x}') &= c_\sigma \mathbb{E}_{(u,v) \sim \mathcal{N}(0, \Lambda^{(\ell)})}[\sigma(u)\sigma(v)]
\end{aligned}
\end{equation}
where $\sigma$ is the activation function and $c_\sigma$ is a normalization constant (e.g., $c_\sigma = 2$ for ReLU activation). The matrix $\Lambda^{(\ell)}$ captures the covariance structure at layer $\ell$, and the expectation computes the kernel value through the activation function.

As network width $m \to \infty$, the NTK converges to a deterministic kernel $\Theta_\infty(\mathbf{x}, \mathbf{x}')$ that depends only on architecture and inputs, not on random initialization \citep{jacot2018neural}. This convergence follows from the law of large numbers, with fluctuations vanishing as $\mathcal{O}(1/\sqrt{m})$. In this regime, the kernel remains approximately constant during training \citep{chizat2018global}.

In the deterministic NTK regime, kernel methods can fully represent infinitely wide neural networks \citep{jacot2018neural}. The function learned by gradient descent becomes exactly equivalent to kernel regression: $f(\mathbf{x}, t) = f(\mathbf{x}, 0) + \Theta(\mathbf{x}, \mathbf{X})\Theta(\mathbf{X}, \mathbf{X})^{-1}(\mathbf{y} - f(\mathbf{X}, 0))(1 - e^{-\eta t})$, where $\mathbf{X}$ is training data, $\eta$ is learning rate and $t$ is the scaled training iterations \citep{arora2019exactcomputationinfinitelywide, Lee_2020}. The NTK inherits neural networks' universal approximation properties \citep{hornik1989multilayer} while operating in the infinite-dimensional tangent space at initialization. Kernel regression with NTK offers convex optimization with unique global optima and closed-form solutions: $\mathbf{f} = \Theta(\mathbf{X}_{test}, \mathbf{X}_{train})(\Theta(\mathbf{X}_{train}, \mathbf{X}_{train}) + \lambda \mathbf{I})^{-1}\mathbf{y}$. 

The NTK captures non-linear relationships through its hierarchical construction—deeper networks yield more expressive kernels capable of modeling complex patterns. Unlike hand-crafted kernels that encode specific assumptions (e.g., smoothness for Gaussian kernels, polynomial interactions for polynomial kernels), the NTK automatically learns appropriate feature representations through its architectural inductive bias. This makes it particularly suitable for genetic data where the optimal similarity measure between genotype profiles is unknown a priori. The NTK preserves deep learning's inductive biases—depth creates compositional learning through recursive structure, while architecture choices are encoded in the kernel.

This formulation enables us to leverage the expressiveness of deep neural networks while maintaining the theoretical guarantees and computational tractability of kernel methods—crucial for genetic risk prediction where both accuracy and reliability are essential.

\subsection{Linear Mixed Models and Kernel Ridge Regression}

Linear Mixed Models and Kernel Ridge Regression represent two fundamental approaches for prediction with structured data \citep{shawe2004kernel,yang2010common}. Consider a response variable $\mathbf{y} \in \mathbb{R}^n$ and a kernel matrix $\mathbf{K} \in \mathbb{R}^{n \times n}$ computed from features $\mathbf{X} \in \mathbb{R}^{n \times p}$, where in genomics applications, $\mathbf{X}$ typically represents genotype data with $n$ individuals and $p$ genetic variants.

In genetic risk prediction, LMMs have emerged as the dominant framework for polygenic score estimation and genome-wide association studies \citep{yang2010common,kang2010variance,yang2011gcta}. The key insight is modeling genetic similarity through kernel functions that capture the aggregate effects of many genetic variants. The most widely used is the linear kernel, also called the genomic relationship matrix (GRM):
\begin{equation}
\mathbf{K}_{ij} = \frac{1}{p}\sum_{k=1}^{p} \frac{(X_{ik} - \mu_k)(X_{jk} - \mu_k)}{\sigma_k^2}
\end{equation}
where $X_{ik}$ is the genotype of individual $i$ at variant $k$, and $\mu_k$, $\sigma_k$ are the mean and standard deviation of variant $k$. This kernel measures genetic similarity as the correlation of standardized genotypes across all variants, capturing the intuition that individuals with similar genetic profiles should have similar phenotypic outcomes. The normalization by $p$ ensures the kernel scales appropriately with the number of variants, while the standardization accounts for different allele frequencies across the genome.

The LMM formulation models the response as:
\begin{equation}
\mathbf{y} = \mathbf{Z}\boldsymbol{\beta} + \mathbf{u} + \boldsymbol{\epsilon}
\end{equation}
where $\mathbf{Z} \in \mathbb{R}^{n \times q}$ is the fixed effects design matrix containing covariates such as age, sex, population principal components, and other demographic variables that must be adjusted for in genetic analyses. The random genetic effects $\mathbf{u} \sim \mathcal{N}(0, \sigma_g^2\mathbf{K})$ model the aggregate contribution of genetic variants, while $\boldsymbol{\epsilon} \sim \mathcal{N}(0, \sigma_e^2\mathbf{I})$ represents residual error. Here, $\sigma_g^2$ is the genetic variance (heritability component) and $\sigma_e^2$ is the environmental variance. The kernel $\mathbf{K}$ encodes our assumption about genetic similarity—the product kernel assumes additive effects across variants \citep{yang2010common}, while polynomial kernels can capture epistatic interactions \citep{akdemir2015locally}.

The variance components $\sigma_g^2$ and $\sigma_e^2$ are unknown parameters that must be estimated from the data before computing predictions. While Restricted Maximum Likelihood is widely used for variance component estimation \citep{patterson1971recovery}, Minimum Norm Quadratic Unbiased Estimation offers superior computational efficiency, particularly for large-scale genetic data \citep{rao1970estimation,rao1971estimation,rao1972estimation}. MINQUE estimates variance components by solving:
\begin{equation}
\hat{\boldsymbol{\theta}} = \mathbf{C}^{-1}\mathbf{u}
\end{equation}
where $\boldsymbol{\theta} = [\sigma_e^2, \sigma_g^2]^T$ are the variance components, $\mathbf{u}_i = \mathbf{y}^T\mathbf{P}\mathbf{V}_i\mathbf{P}\mathbf{y}$ with $\mathbf{V}_0 = \mathbf{I}$ and $\mathbf{V}_1 = \mathbf{K}$, and $\mathbf{C}_{ij} = \text{tr}(\mathbf{P}\mathbf{V}_i\mathbf{P}\mathbf{V}_j)$. The matrix $\mathbf{P} = \mathbf{V}^{-1} - \mathbf{V}^{-1}\mathbf{Z}(\mathbf{Z}^T\mathbf{V}^{-1}\mathbf{Z})^{-1}\mathbf{Z}^T\mathbf{V}^{-1}$ projects out fixed effects, where $\mathbf{V} = \sigma_g^2\mathbf{K} + \sigma_e^2\mathbf{I}$ is initialized using prior weights (MINQUE0) or iteratively updated (MINQUE1). The key computational advantage is that MINQUE requires only solving linear systems rather than iterative optimization, making it tractable for high-dimensional genetic data.

Once variance components are estimated, the Best Linear Unbiased Prediction (BLUP) for the random effects is:
\begin{equation}
\hat{\mathbf{u}} = \hat{\sigma}_g^2\mathbf{K}(\hat{\sigma}_g^2\mathbf{K} + \hat{\sigma}_e^2\mathbf{I})^{-1}(\mathbf{y} - \mathbf{Z}\hat{\boldsymbol{\beta}})
\end{equation}

In contrast to LMM's variance component estimation, KRR offers a more straightforward computational approach \citep{shawe2004kernel}. The KRR solution minimizes the regularized loss:
\begin{equation}
\hat{\boldsymbol{\alpha}} = \arg\min_{\boldsymbol{\alpha}} \|\mathbf{y} - \mathbf{K}\boldsymbol{\alpha}\|^2 + \lambda \boldsymbol{\alpha}^T\mathbf{K}\boldsymbol{\alpha}
\end{equation}
yielding the closed-form solution $\hat{\boldsymbol{\alpha}} = (\mathbf{K} + \lambda\mathbf{I})^{-1}\mathbf{y}$. The regularization parameter $\lambda$ is typically selected through cross-validation, testing a grid of values (e.g., $\lambda \in \{10^{-3}, 10^{-2}, \ldots, 10^{2}\}$) and choosing the one minimizing prediction error on held-out data. This cross-validation approach is computationally simple—requiring only matrix inversions for each $\lambda$ candidate—and avoids the complex variance component estimation procedures required by LMMs. For prediction on new samples with kernel values $\mathbf{k}_*$, we simply compute $\hat{y}_* = \mathbf{k}_*^T\hat{\boldsymbol{\alpha}}$, making KRR particularly attractive when prediction accuracy is the primary goal rather than heritability estimation.

The relationship between KRR and LMM reveals complementary strengths for genetic prediction. When $\lambda = \sigma_e^2/\sigma_g^2$, KRR and LMM predictions become mathematically equivalent, demonstrating that the regularization parameter $\lambda$ directly corresponds to the noise-to-signal ratio estimated through variance components. This equivalence highlights a fundamental trade-off between computational simplicity and statistical richness. KRR offers a more straightforward approach through cross-validation for selecting $\lambda$, directly optimizing predictive performance without requiring variance component estimation. This computational efficiency and direct focus on prediction accuracy often makes KRR the preferred choice when the primary goal is maximizing predictive power. However, the LMM framework provides crucial advantages for comprehensive genetic analysis. LMMs naturally accommodate both fixed effects (population structure, age, sex, batch effects) and random genetic effects within a unified model: $\mathbf{y} = \mathbf{Z}\boldsymbol{\beta} + \mathbf{u} + \boldsymbol{\epsilon}$, with proper uncertainty quantification for all components. More importantly, the LMM framework extends beyond prediction to enable formal statistical inference—kernel-based association tests like SKAT leverage the mixed model structure to test whether sets of genetic variants are associated with phenotypes while controlling for confounders \citep{wu2011rare}. This inferential capability, combined with interpretable variance component decomposition that quantifies heritability, makes LMMs invaluable when understanding biological mechanisms is as important as prediction accuracy. By embedding advanced kernels like the NTK within both frameworks, we can choose between KRR's computational efficiency for pure prediction tasks and LMM's comprehensive statistical framework for biological discovery, selecting the approach that best matches our scientific objectives.

\subsection{NTK-Based Methods for Genetic Risk Prediction}

We propose two complementary approaches that leverage the NTK for genetic risk prediction. Both methods use the NTK to capture complex, non-linear relationships between genetic variants while maintaining computational efficiency. In practice, computing the limiting NTK, $\Theta_\infty$, as network width approaches infinity is intractable. Instead, we employ a sufficiently wide neural network (e.g., hidden dimension $m \geq 1000$) and compute the empirical NTK using the network's gradients at initialization:
\begin{equation}
\Theta_{\text{empirical}}(\mathbf{x}, \mathbf{x}') = \sum_{i} \frac{\partial f(\mathbf{x}; \boldsymbol{\theta}_0)}{\partial \theta_i} \frac{\partial f(\mathbf{x}'; \boldsymbol{\theta}_0)}{\partial \theta_i}
\end{equation}
where $\boldsymbol{\theta}_0$ are the initial network parameters. For wide networks, this empirical NTK closely approximates the deterministic limiting kernel, with deviations of order $\mathcal{O}(1/\sqrt{m})$. We treat this empirical NTK as our deterministic kernel for all subsequent computations.

\textbf{\emph{NTK-LMM} Approach:} We integrate the NTK into the LMM framework for genetic prediction. Given genotype matrix $\mathbf{X} \in \mathbb{R}^{n \times p}$ containing $n$ individuals and $p$ variants, we compute the NTK matrix $\mathbf{K}_{\text{NTK}}$ using a wide neural network with hidden dimension $m$ chosen based on the input dimension. The variance components $\sigma_g^2$ and $\sigma_e^2$ are estimated using MINQUE as described earlier, solving the linear system:
\begin{equation}
\begin{bmatrix} \hat{\sigma}_e^2 \\ \hat{\sigma}_g^2 \end{bmatrix} = \mathbf{C}^{-1}\mathbf{u}
\end{equation}
where $\mathbf{C}_{ij} = \text{tr}(\mathbf{P}\mathbf{V}_i\mathbf{P}\mathbf{V}_j)$ and $\mathbf{u}_i = \mathbf{y}^T\mathbf{P}\mathbf{V}_i\mathbf{P}\mathbf{y}$ with $\mathbf{V}_0 = \mathbf{I}$ and $\mathbf{V}_1 = \mathbf{K}_{\text{NTK}}$. The prediction for new individuals follows the BLUP framework:
\begin{equation}
\begin{aligned}
&\hat{\mathbf{y}}_{\text{test}} = \mathbf{Z}_{\text{test}}\hat{\boldsymbol{\beta}} \\&+ \mathbf{K}_{\text{NTK,test-train}}(\mathbf{K}_{\text{NTK,train}} + \hat{\lambda}\mathbf{I})^{-1}(\mathbf{y}_{\text{train}} - \mathbf{Z}_{\text{train}}\hat{\boldsymbol{\beta}})
\end{aligned}
\end{equation}
where $\hat{\lambda} = \hat{\sigma}_e^2/\hat{\sigma}_g^2$ is derived from the estimated variance components, providing interpretability through heritability estimates.

\textbf{\emph{NTK-KRR} Approach:} This method directly applies KRR with the NTK, bypassing variance component estimation in favor of predictive optimization. Using the same NTK computation, we select the regularization parameter through $k$-fold cross-validation over a grid $\Lambda = \{10^{-3}, 10^{-2}, \ldots, 10^{2}\}$:
\begin{equation}
\begin{aligned}
\lambda^* =& \arg\min_{\lambda \in \Lambda}\\& \frac{1}{k}\sum_{i=1}^{k} \|\mathbf{y}_i^{\text{val}} - \mathbf{K}_{\text{NTK}}^{(i)}(\mathbf{K}_{\text{NTK,train}}^{(i)} + \lambda\mathbf{I})^{-1}\mathbf{y}_{\text{train}}^{(i)}\|^2
\end{aligned}
\end{equation}
where superscript $(i)$ denotes the $i$-th fold partition. The final prediction uses:
\begin{equation}
\hat{\mathbf{y}}_{\text{test}} = \mathbf{K}_{\text{NTK,test-train}}(\mathbf{K}_{\text{NTK,train}} + \lambda^*\mathbf{I})^{-1}\mathbf{y}_{\text{train}}
\end{equation}

The key distinction between these approaches lies in their optimization objectives: \emph{NTK-LMM} estimates variance components to decompose phenotypic variance into genetic and environmental contributions, enabling heritability analysis alongside prediction. \emph{NTK-KRR} prioritizes pure predictive accuracy through cross-validation, offering computational simplicity at the cost of interpretability. Both methods leverage the NTK's ability to capture hierarchical feature interactions while maintaining the theoretical guarantees of kernel methods.

\subsection{Computational Complexity}

In terms of computational complexity, the main difference between NTK-based methods and traditional neural networks lies in how the costs are distributed. Training a traditional neural network requires $\mathcal{O}(E\,nP)$ time, where $E$ is the number of epochs, $n$ the number of samples, and $P$ the number of parameters, with memory dominated by $\mathcal{O}(P)$. In contrast, NTK methods avoid iterative training but front-load computation into two steps: constructing the empirical NTK, which requires $\mathcal{O}(n^2P)$ time and $\mathcal{O}(n^2)$ memory, and solving a single kernel system, which costs $\mathcal{O}(n^3)$ time and $\mathcal{O}(n^2)$ memory. For \emph{NTK-LMM}, variance component estimation via MINQUE also reduces to a handful of $\mathcal{O}(n^3)$ linear solves. Standard LMMs, by comparison, require $\mathcal{O}(n^2p)$ to construct the genomic relationship matrix followed by $\mathcal{O}(n^3)$ inference, with similar $\mathcal{O}(n^2)$ memory usage. 

Overall, neural networks spread their computational cost across many training epochs, while NTK-based methods incur a heavier upfront kernel construction but then rely on convex, closed-form optimization. Compared to classical LMMs, NTK adds the $\mathcal{O}(n^2P)$ kernel-building step but provides greater modeling flexibility and stability without iterative optimization. This trade-off highlights NTK as a computationally efficient alternative in settings where $E$ is large or stable training of deep networks is challenging.

\section{Simulation Studies}
\begin{figure*}[t]
  \centering
  \vspace{.3in}
  \includegraphics[width=\textwidth]{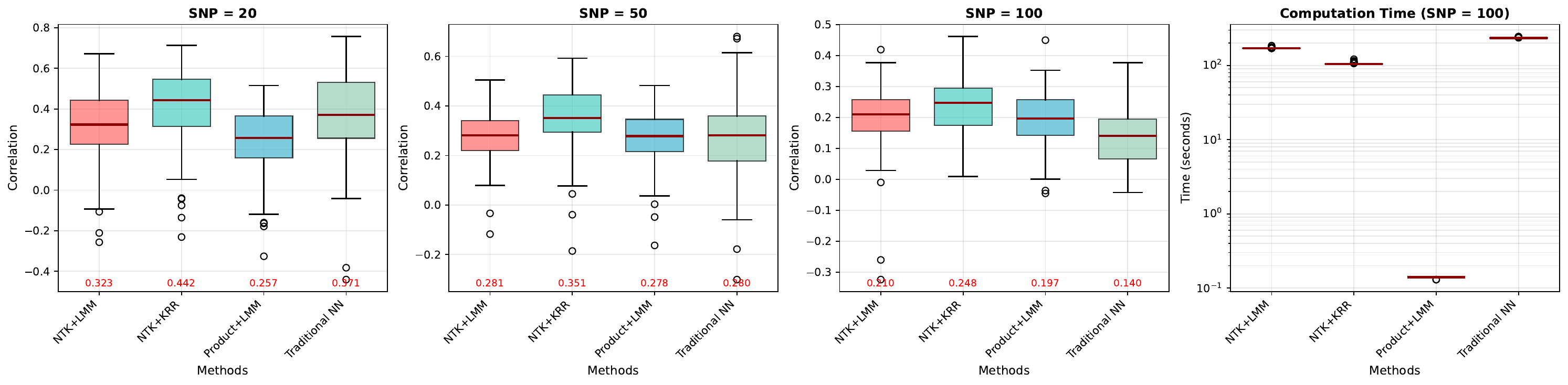}
  \vspace{.3in}
  \caption{Testing correlation and computational time for Ricker Curve Model scenario.}
  \label{Simulation}
\end{figure*}
We tested our NTK-based methods through simulations using real genetic data to ensure realistic results. We randomly selected $n = 1000$ people from the UK Biobank database \citep{sudlow2015uk} and, for each simulation run, randomly picked $p$ genetic markers (SNPs), where $p \in \{20, 50, 100\}$. We only used common genetic variants (those appearing in more than 5\% of people, technically MAF $> 0.05$) to ensure enough variation for analysis. This gave us a genotype matrix $\mathbf{X} \in \mathbb{R}^{n \times p}$ containing real-world genetic patterns—including how certain genes are inherited together (linkage disequilibrium) and natural population differences (population stratification)—that wouldn't appear in artificially generated data under Hardy-Weinberg equilibrium \citep{pasaniuc2017dissecting}. Finally, we centered the genetic data around its mean to remove population-wide trends while keeping individual differences intact \citep{yang2011gcta}.

\subsection{Simulation Design}
We first created genetic signals $\mathbf{g}$ and random noise $\boldsymbol{\epsilon}$ using:
\begin{equation}
    \mathbf{g} \sim \mathcal{N}\left(\mathbf{0},\frac{\sigma_g^2}{p} \mathbf{X} \mathbf{X}^T\right),\quad \boldsymbol{\epsilon} \sim \mathcal{N}\left(\mathbf{0}, \sigma_e^2\mathbf{I}\right)
\end{equation}
where $\mathbf{g}$ represents the genetic contribution and $\boldsymbol{\epsilon}$ represents random environmental effects.

We tested five different models to see how our methods handle various types of genetic relationships:

\textbf{Linear Model:} The simplest case where traits are directly proportional to genetic values:
\begin{equation}
\mathbf{y} = \mathbf{g} + \boldsymbol{\epsilon}
\end{equation}

\textbf{Hyperbolic Model:} Models saturation effects where genetic effects level off at extreme values:
\begin{equation}
\mathbf{y} = \frac{r(\mathbf{g}^\alpha)}{1+r(\mathbf{g}^\alpha)} + \boldsymbol{\epsilon}
\end{equation}
where $r(x) = \log(1 + e^x)$ is the softplus function and $\alpha$ controls how curved the relationship is.

\textbf{Power Model:} Captures scenarios where genetic effects are amplified non-linearly:
\begin{equation}
\mathbf{y} = \mathbf{g}^\alpha + \boldsymbol{\epsilon}
\end{equation}
where $\alpha > 1$ makes strong genetic effects even stronger.

\textbf{Cosh Model:} Creates symmetric non-linear effects where both increases and decreases from average have amplified impacts:
\begin{equation}
\mathbf{y} = \cosh(\mathbf{g}) + \boldsymbol{\epsilon} = \frac{\exp(\mathbf{g})+\exp(-\mathbf{g})}{2} + \boldsymbol{\epsilon}
\end{equation}

\textbf{Ricker Curve Model:} Represents relationships where intermediate genetic values are optimal (like height or weight):
\begin{equation}
\mathbf{y} = r(\mathbf{g}^\alpha) \exp[-r(\mathbf{g}^\alpha)] + \boldsymbol{\epsilon}
\end{equation}
where the curve peaks at intermediate values, modeling how extreme genetic values in either direction can be disadvantageous. And again $r(x) = \log(1 + e^x)$ is the softplus function.

We ran each scenario 100 times, splitting the data 80\% for training and 20\% for testing in each run. All experiments were run on the University of Florida HiPerGator cluster using NVIDIA B200 GPUs with CUDA acceleration. GPU resources were used for empirical NTK construction and neural network baselines; kernel assembly and linear solves (e.g., KRR/LMM systems and MINQUE components) were executed on CPU nodes with optimized BLAS/LAPACK. This configuration provided sufficient memory bandwidth and throughput to handle datasets with up to thousands of samples and hundreds of features in our runs.

\subsection{Baseline Methods and Implementation}
We compared four methods spanning traditional statistical genetics and modern machine learning approaches: \textbf{(1) \emph{NTK-LMM}:} Combines the NTK with linear mixed model framework, estimating variance components $\sigma_g^2$ and $\sigma_e^2$ via MINQUE before computing BLUP predictions. \textbf{(2) \emph{NTK-KRR}:} Direct kernel ridge regression using the NTK, with regularization parameter selected through 5-fold cross-validation. \textbf{(3) Product LMM:} The standard genomic BLUP approach using the genomic relationship matrix $\mathbf{K}_{ij} = \frac{1}{p}\sum_{k=1}^{p} X_{ik}X_{jk}$, representing current best practice in genetic prediction. \textbf{(4) Traditional Neural Network:} A fully-connected network with architecture [input $\rightarrow$ 50 $\rightarrow$ 30 $\rightarrow$ output], using ReLU activations, batch normalization, and dropout (rate 0.2), trained via backpropagation.

The NTK computation employed a fully-connected network at initialization with adaptive architecture: for $p < 50$, we used width $m = 2000$ and depth 2, while for $p \geq 50$, we used width $m = 1000$ and depth 3, balancing expressiveness with computational tractability. These configurations ensure the empirical NTK closely approximates the infinite-width limit while remaining computationally feasible. The traditional neural network was trained for up to 200 epochs using Adam optimizer ($\eta = 10^{-3}$), L2 weight decay ($\lambda = 10^{-3}$), batch size 32, and early stopping with patience of 20 epochs based on validation loss (10\% validation split). For \emph{NTK-KRR}, we selected the regularization parameter $\alpha^* \in \{10^{-3}, 10^{-2}, 10^{-1}, 1, 10\}$ via 5-fold cross-validation, choosing the value minimizing mean squared error on held-out folds.

\subsection{Results and Interpretation}

The simulation results reveal differences in how methods scale with increasing dimensionality across all tested genetic architectures (comprehensive results in Supplementary Materials). Using the Ricker curve scenario as a representative example (Figure \ref{Simulation}), at SNP=20, \emph{NTK-KRR} achieves the highest median correlation (0.442), followed by traditional neural networks (0.371), \emph{NTK-LMM} (0.323), and product LMM (0.257). This pattern—where methods capturing non-linear relationships substantially outperform linear approaches in manageable feature spaces—remains consistent across the linear, hyperbolic, power, and cosh scenarios, though the magnitude of improvement varies with the complexity of the underlying genetic architecture.

The scaling behavior exposes fundamental robustness differences. As dimensionality increases to SNP=100, traditional neural networks degrade from 0.371 to 0.140 (62\% reduction), while \emph{NTK-KRR} maintains 0.248 (44\% reduction). This collapse stems from traditional neural networks' vulnerability to high-dimensional optimization: the parameter space grows quadratically, creating complex loss landscapes with numerous local minima.

The product LMM exhibits the most stable scaling, with performance declining only 23\% from SNP=20 to SNP=100. This robustness reflects principles of quantitative genetics, as the genomic relationship matrix approaches a near-diagonal form when the number of variants increases. Nevertheless, the inherent linearity of LMMs limits their ability to capture non-linear effects, which explains their consistent underperformance relative to NTK-based methods.

NTK approaches achieve optimal balance between expressiveness and stability. \emph{NTK-KRR} consistently outperforms all methods, while \emph{NTK-LMM} shows intermediate performance, surpassing product kernel LMM but with less dimensional robustness than \emph{NTK-KRR}. The superiority of \emph{NTK-KRR} over \emph{NTK-LMM} formulation likely reflects direct optimization through cross-validation versus potential misspecification in MINQUE variance component estimation.

Computational efficiency analysis reveals important trade-offs. Traditional neural networks require 233 seconds at SNP=100 for 10,000 epochs, while \emph{NTK-LMM} and \emph{NTK-KRR} complete in 167 and 104 seconds respectively. In contrast, product kernel LMM achieves near-instantaneous computation at 0.14 seconds through closed-form solutions. While LMM offers unmatched speed, NTK methods provide a practical compromise—1.4-2.2× faster than neural networks while maintaining superior accuracy and dimensional robustness, making them particularly suitable for genomic applications where both performance and scalability matter.

\section{Real Data Analysis}

\begin{figure*}[t]
  \centering
  \vspace{.3in}
  \includegraphics[width=0.8\textwidth]{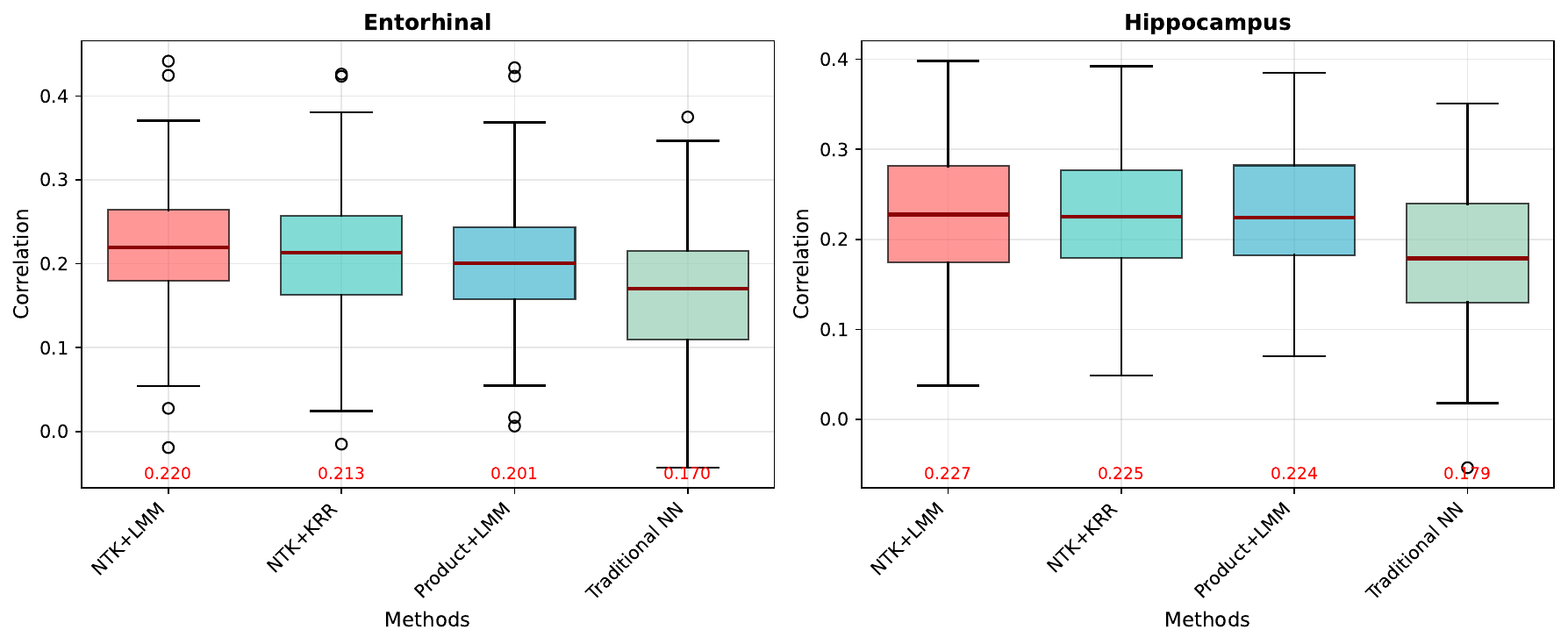}
  \vspace{.3in}
  \caption{Testing correlation for Entorhinal and Hippocampus outcomes.}
  \label{Realdata}
\end{figure*}

\subsection{ADNI Dataset and Preprocessing}

We evaluated our methods on the Alzheimer's Disease Neuroimaging Initiative dataset \citep{saykin2010alzheimer}, focusing on the APOE-TOMM40-APOC1 gene cluster—a well-established locus for Alzheimer's disease risk and related endophenotypes \citep{karch2015alzheimer}. This genomic region provides an ideal test case as it contains variants with known strong effects on neurodegeneration, allowing us to assess whether NTK methods can capture these established biological signals.

We applied standard quality control filters: minor allele frequency $\geq 0.01$ and Hardy-Weinberg equilibrium p-value $\geq 10^{-6}$. Data with 738 SNPs spanning the locus and 697 individuals is used for further analysis. We examined two quantitative neuroimaging outcomes strongly linked to Alzheimer's progression: entorhinal cortex thickness and hippocampal volume. To account for demographic confounders, we adjusted both outcomes for sex, ethnicity, and education level as fixed effects, then standardized the residuals to zero mean.

To ensure robust performance assessment, we employed 100 random 80-20 train-test splits, allowing us to evaluate both prediction accuracy and stability across different data partitions. This repeated sampling strategy provides confidence intervals for our performance metrics and guards against overfitting to particular train-test divisions.

\subsection{Results and Clinical Interpretation}
The real data analysis highlighted distinct patterns across brain regions, consistent with but more nuanced than the simulation findings.  

For the \textbf{entorhinal cortex}, a critical region for early Alzheimer’s pathology\citep{donix2013apoe}, all methods demonstrated moderate predictive ability. \emph{NTK-LMM} achieved the highest mean correlation (0.220), closely followed by \emph{NTK-KRR} (0.213), product kernel LMM (0.201), and the traditional neural network (0.170). The advantage of kernel-based approaches over the neural network underscores their ability to capture subtle non-linear effects of APOE-region variants on entorhinal structure. The relatively narrow spread of performance suggests that genetic contributions to entorhinal variation are detectable but modest.  

In the \textbf{hippocampus}, performance across methods was even more uniform, with \emph{NTK-LMM} (0.227), \emph{NTK-KRR} (0.225), and product kernel LMM (0.224) yielding nearly identical results, and the neural network trailing slightly (0.179). This convergence indicates that, for hippocampal volume, the APOE region contributes only limited variance, and most methods saturate at the same predictive ceiling. This finding aligns with prior evidence that hippocampal atrophy in Alzheimer’s disease reflects polygenic and environmental influences beyond APOE alone \citep{altmann2020comprehensive}.  

These findings have important implications for precision medicine in Alzheimer's disease. The improved prediction accuracy of NTK-based methods could enhance early risk stratification, particularly for individuals with complex genetic profiles not well-captured by linear polygenic scores. Furthermore, the ability to model non-linear genetic effects may reveal novel biological mechanisms underlying neurodegeneration, guiding therapeutic target discovery.

\section{Conclusion}

We introduced \emph{NTK-LMM} and \emph{NTK-KRR}, embedding the Neural Tangent Kernel within statistical genetics frameworks for genetic risk prediction. Our evaluation demonstrated that NTK methods maintain state-of-the-art performance as dimensionality increases—while traditional neural networks show degradation at higher SNP counts, NTK methods exhibit substantially better robustness. The product kernel LMM demonstrated the most stable scaling but consistently underperformed NTK methods due to its inability to capture non-linear effects. ADNI analysis validated these findings, with NTK methods achieving highest accuracy for entorhinal cortex thickness, though all methods converged for hippocampal volume—suggesting performance gains from non-linear modeling are trait-dependent.

The mathematical equivalence between KRR and LMM when $\lambda = \sigma_e^2/\sigma_g^2$ illuminates our approaches' complementary nature. \emph{NTK-KRR} offers computational simplicity through cross-validation, ideal for maximizing predictive accuracy. \emph{NTK-LMM} preserves the mixed model framework, enabling heritability estimation and supporting future development of formal inference procedures. Just as SKAT extends LMMs for association testing \citep{wu2011rare,hou2023association}, similar frameworks could be developed for \emph{NTK-LMM}, enabling tests of complex non-linear associations with rigorous statistical control.

Computationally, NTK methods balance efficiency and flexibility—notably faster than neural networks while avoiding iterative optimization challenges. This positions NTK as a practical bridge between deep learning and kernel methods, combining the expressiveness needed for complex genetic architectures with the stability and interpretability required for biomedical applications.

\bibliographystyle{plainnat} 
\bibliography{sample_paper}

\section*{Supplementary Materials}
\subsection{Network Architecture and Parameterization}

We derive the Neural Tangent Kernel (NTK) corresponding to the architecture and implementation in \texttt{KNN.py}, and we justify the empirical calculation procedure. Consider a fully connected neural network 
$f(\mathbf{x}; \boldsymbol{\theta}): \mathbb{R}^p \to \mathbb{R}$ 
of depth $L$ and hidden width $m \to \infty$. The activation function is the rectified linear unit 
$\phi(z) = \max(0,z)$. 
Weights follow NTK scaling: $\mathbf{W}^{(\ell)}\sim\mathcal{N}(0,1)$ for $\ell=1,\dots,L-1$, $\mathbf{w}^{(L)}\sim\mathcal{N}(0,1)$; biases are zero. After each hidden ReLU we multiply by $\sqrt{2/m}$, and the output is scaled by $1/\sqrt{m}$. The recursive definition of the hidden states is

$$
\mathbf{h}^{(0)}(\mathbf{x}) = \mathbf{x},
$$

$$
\mathbf{h}^{(\ell)}(\mathbf{x}) = \sqrt{\tfrac{2}{m}} \, 
\phi\!\left( \mathbf{W}^{(\ell)} \mathbf{h}^{(\ell-1)}(\mathbf{x}) \right), 
\quad \ell = 1,\dots,L-1,
$$

and the output layer is

$$
f(\mathbf{x}; \boldsymbol{\theta}) 
= \tfrac{1}{\sqrt{m}} \, \mathbf{w}^{(L)} \cdot \mathbf{h}^{(L-1)}(\mathbf{x}).
$$

This scaling and initialization match the implementation in our code and ensure that signals remain numerically stable as $m$ grows large.

\subsection{Recursive NTK Formulation}

The empirical implementation normalizes the kernel to have approximately unit diagonal after construction; nevertheless, the theoretical base case is written in terms of the input covariance. Specifically, the initial covariance is

$$
\Sigma^{(0)}(x,x') = \langle \mathbf{x}, \mathbf{x}' \rangle,
$$

without the $1/p$ factor used in some normalizations. For each layer $\ell \geq 1$, define

$$
\Sigma^{(\ell)}(x,x') = 
\mathbb{E}_{(u,v) \sim \mathcal{N}(0, \boldsymbol{\Lambda}^{(\ell-1)})}
[ \phi(u)\phi(v) ],
$$

$$
\dot{\Sigma}^{(\ell)}(x,x') = 
\mathbb{E}_{(u,v) \sim \mathcal{N}(0, \boldsymbol{\Lambda}^{(\ell-1)})}
[ \phi'(u)\phi'(v) ],
$$

where

$$
\boldsymbol{\Lambda}^{(\ell-1)} =
\begin{pmatrix}
\Sigma^{(\ell-1)}(x,x) & \Sigma^{(\ell-1)}(x,x') \\
\Sigma^{(\ell-1)}(x',x) & \Sigma^{(\ell-1)}(x',x')
\end{pmatrix}.
$$

For ReLU activation, the closed forms are

$$
\Sigma^{(\ell)}(x,x') 
= \frac{1}{2\pi}\,
\sqrt{\Sigma^{(\ell-1)}(x,x)\,\Sigma^{(\ell-1)}(x',x')}
\left( \sin \theta + (\pi - \theta)\cos\theta \right),
$$

$$
\dot{\Sigma}^{(\ell)}(x,x') = \frac{\pi - \theta}{2\pi}, 
\qquad
\theta = \cos^{-1}\!\Bigg( 
\tfrac{\Sigma^{(\ell-1)}(x,x')}
{\sqrt{\Sigma^{(\ell-1)}(x,x)\,\Sigma^{(\ell-1)}(x',x')}} 
\Bigg).
$$

The NTK recursion is

$$
\Theta^{(\ell)}(x,x') = 
\Theta^{(\ell-1)}(x,x') \cdot \dot{\Sigma}^{(\ell)}(x,x') 
+ \Sigma^{(\ell)}(x,x'),
$$

with $\Theta^{(0)}(x,x') = \Sigma^{(0)}(x,x')$. This recursion describes the deterministic NTK $\Theta_\infty$ that arises in the infinite-width limit, up to the additional scaling factors included in the implementation.

\subsection{Empirical Computation}

In practice, the NTK is computed empirically using automatic differentiation in PyTorch. 
For each block of samples $(\mathbf{x}_i, \mathbf{x}_j)$, we evaluate the Jacobian of the 
network output with respect to all parameters 
$\boldsymbol{\theta} = (\mathbf{W}^{(1)}, \dots, \mathbf{W}^{(L)}, \mathbf{w}^{(L)})$ 
at a single random initialization $\boldsymbol{\theta}_0$, and compute

$$
\hat{\Theta}(\mathbf{x}_i, \mathbf{x}_j) 
= \langle 
\nabla_{\boldsymbol{\theta}} f_{\boldsymbol{\theta}_0}(\mathbf{x}_i), 
\nabla_{\boldsymbol{\theta}} f_{\boldsymbol{\theta}_0}(\mathbf{x}_j)
\rangle,
$$

where the inner product runs over all parameters. The implementation proceeds in blocks to avoid memory overflow: Jacobians are computed for subsets of samples and accumulated to form the full kernel matrix.

After construction, several post-processing steps are applied to improve numerical stability:  
(i) the kernel is symmetrized,  
(ii) the diagonal is rescaled so that $\hat{\Theta}(x_i,x_i)=1$ for all $i$, and  
(iii) the entire kernel is divided by the number of features $p$ to keep scale comparable 
with standard genomic kernels.  

Unlike the theoretical formulation where one may average across multiple random initializations to approximate the infinite-width NTK, our default implementation uses a single initialization ($M=1$). In practice, this yields stable and reproducible results when combined with large hidden width (e.g., $m=2000$) and depth $L=3$, which are consistent with the theoretical assumptions. The resulting positive semidefinite kernel is then used in downstream linear mixed model analysis.

\subsection{Convergence Analysis}

\begin{theorem}[Convergence to Deterministic NTK]
For a fully connected ReLU network with width $m$ per hidden layer, let $\Theta_m(\mathbf{x}, \mathbf{x}')$ denote the empirical NTK computed at initialization. Then, for any $\epsilon > 0$,
$$
\mathbb{P}\!\left(\big|\Theta_m(\mathbf{x}, \mathbf{x}') - \Theta_\infty(\mathbf{x}, \mathbf{x}')\big| > \epsilon\right) 
\leq 2\exp\!\left(-\frac{c\,m\epsilon^2}{L^2}\right),
$$
where $c > 0$ is a constant depending on the depth $L$ and the input norms.
\end{theorem}

\begin{proof}
At each layer $\ell$, the empirical kernel contribution can be expressed as
$$
\Theta^{(\ell)}_m = \frac{1}{m} \sum_{i=1}^m Z_i,
$$
where $Z_i$ are i.i.d. random variables determined by the Gaussian-initialized weights and the ReLU activations. Since the weights are sampled from sub-Gaussian distributions, each $Z_i$ is sub-Gaussian with variance proxy bounded by a constant $C^2$ depending on $L$ and $\|\mathbf{x}\|,\|\mathbf{x}'\|$.  

By standard sub-Gaussian concentration (Chernoff/Hoeffding-type for sub-Gaussians),
$$
\mathbb{P}\!\left(\big|\Theta^{(\ell)}_m - \mathbb{E}[\Theta^{(\ell)}_m]\big| > t\right) 
\leq 2\exp\!\left(-\frac{c\,m t^2}{C^2}\right).
$$
Applying a union bound across $L$ layers and setting $\epsilon = L t$ gives
$$
\mathbb{P}\!\left(\big|\Theta_m(\mathbf{x}, \mathbf{x}') - \Theta_\infty(\mathbf{x}, \mathbf{x}')\big| > \epsilon\right) 
\leq 2\exp\!\left(-\frac{c\,m\epsilon^2}{L^2 C^2}\right).
$$

Thus $\Theta_m(\mathbf{x},\mathbf{x}')$ converges to $\Theta_\infty(\mathbf{x},\mathbf{x}')$ at rate $\mathcal{O}(L/\sqrt{m})$, consistent with Gaussian weight initialization and the concentration behavior observed in practice.
\end{proof}

For our implementation with $m=2000$ and $L=3$, the probability bound implies that deviations of $\Theta_m$ from $\Theta_\infty$ are negligible compared to the stochastic noise inherent in genetic data. This validates that the empirical kernel computed in \texttt{KNN.py} is a reliable approximation to the theoretical NTK.

\begin{proposition}[Equivalence of LMM and KRR Predictions]
Let $\mathbf{y} \in \mathbb{R}^n$ denote the response vector for $n$ training samples, $\mathbf{K} \in \mathbb{R}^{n \times n}$ a positive semidefinite kernel (genomic relationship matrix or NTK), and variance components $(\sigma_g^2, \sigma_e^2)$ define the linear mixed model
$$
\mathbf{y} = \mathbf{u} + \boldsymbol{\epsilon}, 
\qquad \mathbf{u} \sim \mathcal{N}(\mathbf{0}, \sigma_g^2 \mathbf{K}), 
\quad \boldsymbol{\epsilon} \sim \mathcal{N}(\mathbf{0}, \sigma_e^2 \mathbf{I}).
$$
Then the best linear unbiased predictor (BLUP) of the random effect implies that the prediction of $\mathbf{y}_{\mathrm{test}}$ is identical to kernel ridge regression (KRR) with regularization parameter $\lambda = \sigma_e^2/\sigma_g^2$. Formally,
$$
\hat{\mathbf{y}}_{\mathrm{LMM}} 
= \mathbf{K}_{\mathrm{test,train}} 
\left(\mathbf{K}_{\mathrm{train}} + \lambda \mathbf{I}\right)^{-1} 
\mathbf{y}_{\mathrm{train}} 
= \hat{\mathbf{y}}_{\mathrm{KRR}}.
$$
\end{proposition}

\begin{proof}
We start from the linear mixed model formulation. For training samples, the covariance structure of $\mathbf{y}$ is
$$
\mathbf{V} = \mathrm{Var}(\mathbf{y}) = \sigma_g^2 \mathbf{K} + \sigma_e^2 \mathbf{I},
$$
where $\mathbf{K}$ is the $n \times n$ kernel (e.g., genomic relationship or empirical NTK) and $\mathbf{I}$ is the $n \times n$ identity matrix.  

The BLUP of the random effect $\mathbf{u}$ given $\mathbf{y}$ is
$$
\hat{\mathbf{u}} = \mathbb{E}[\mathbf{u} \mid \mathbf{y}] 
= \sigma_g^2 \mathbf{K} \mathbf{V}^{-1}(\mathbf{y} - \mathbf{Z}\hat{\boldsymbol{\beta}}),
$$
where $\mathbf{Z}$ denotes fixed-effect covariates (such as intercepts or additional covariates) and $\hat{\boldsymbol{\beta}}$ is their generalized least squares estimator. In the special case where fixed effects are absent ($\mathbf{Z}=\mathbf{0}$) or orthogonal to the kernel space, this simplifies to
$$
\hat{\mathbf{u}} = \sigma_g^2 \mathbf{K} (\sigma_g^2 \mathbf{K} + \sigma_e^2 \mathbf{I})^{-1} \mathbf{y}.
$$

To obtain predictions for a new test sample $\mathbf{x}_{\mathrm{test}}$, we compute
$$
\hat{\mathbf{y}}_{\mathrm{LMM}} 
= \mathbf{K}_{\mathrm{test,train}} \, \mathbf{V}^{-1} \mathbf{y} \, \sigma_g^2,
$$
where $\mathbf{K}_{\mathrm{test,train}}$ is the vector of kernel evaluations between the test sample and all $n$ training samples. Substituting $\mathbf{V} = \sigma_g^2 \mathbf{K} + \sigma_e^2 \mathbf{I}$ gives
$$
\hat{\mathbf{y}}_{\mathrm{LMM}}
= \mathbf{K}_{\mathrm{test,train}} \, \sigma_g^2 (\sigma_g^2 \mathbf{K} + \sigma_e^2 \mathbf{I})^{-1} \mathbf{y}
= \mathbf{K}_{\mathrm{test,train}} \, (\mathbf{K} + \tfrac{\sigma_e^2}{\sigma_g^2} \mathbf{I})^{-1} \mathbf{y}.
$$

Now compare with kernel ridge regression. In KRR, the training problem solves
$$
\hat{\boldsymbol{\alpha}} 
= \arg\min_{\boldsymbol{\alpha} \in \mathbb{R}^n} 
\|\mathbf{y} - \mathbf{K}\boldsymbol{\alpha}\|^2 
+ \lambda \boldsymbol{\alpha}^\top \mathbf{K} \boldsymbol{\alpha},
$$
where $\lambda > 0$ is the ridge regularization parameter. The solution is
$$
\hat{\boldsymbol{\alpha}} = (\mathbf{K} + \lambda \mathbf{I})^{-1} \mathbf{y}.
$$
Predictions for a new test point are
$$
\hat{\mathbf{y}}_{\mathrm{KRR}} 
= \mathbf{K}_{\mathrm{test,train}} \hat{\boldsymbol{\alpha}} 
= \mathbf{K}_{\mathrm{test,train}} (\mathbf{K} + \lambda \mathbf{I})^{-1} \mathbf{y}.
$$

Identifying $\lambda = \sigma_e^2/\sigma_g^2$ shows that the LMM predictor coincides exactly with the KRR predictor:
$$
\hat{\mathbf{y}}_{\mathrm{LMM}} = \hat{\mathbf{y}}_{\mathrm{KRR}}.
$$

Thus the two methods differ only in interpretation: in LMMs, $\lambda$ arises as the ratio of variance components estimated from data, while in KRR, $\lambda$ is a tuning parameter selected via cross-validation or other criteria.
\end{proof}

\subsection{Additional Simulation Results}

In the main text we summarized simulation results under several distinct data-generating mechanisms but did not include full plots due to space limitations. Here we provide the complete figures for the linear, power, hyperbolic cosine, and hyperbolic scenarios. Each figure reports both predictive accuracy, measured by testing correlation, and computational efficiency, measured by runtime. These additional plots complement the main text by illustrating the consistency of our conclusions across different nonlinear settings and by providing a detailed view of the trade-offs between accuracy and computational cost.

\begin{figure}[H]
  \centering
  \includegraphics[width=\textwidth]{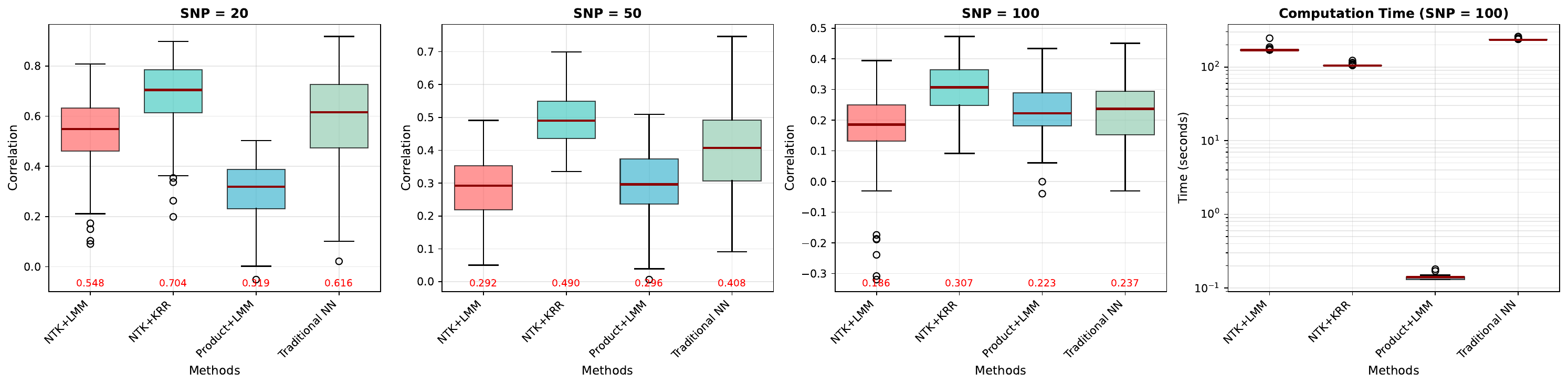}
  \vspace{.2in}
  \caption{Testing correlation and computational time for linear model scenario.}
  \label{Simulation1}
\end{figure}

\begin{figure}[H]
  \centering
  \includegraphics[width=\textwidth]{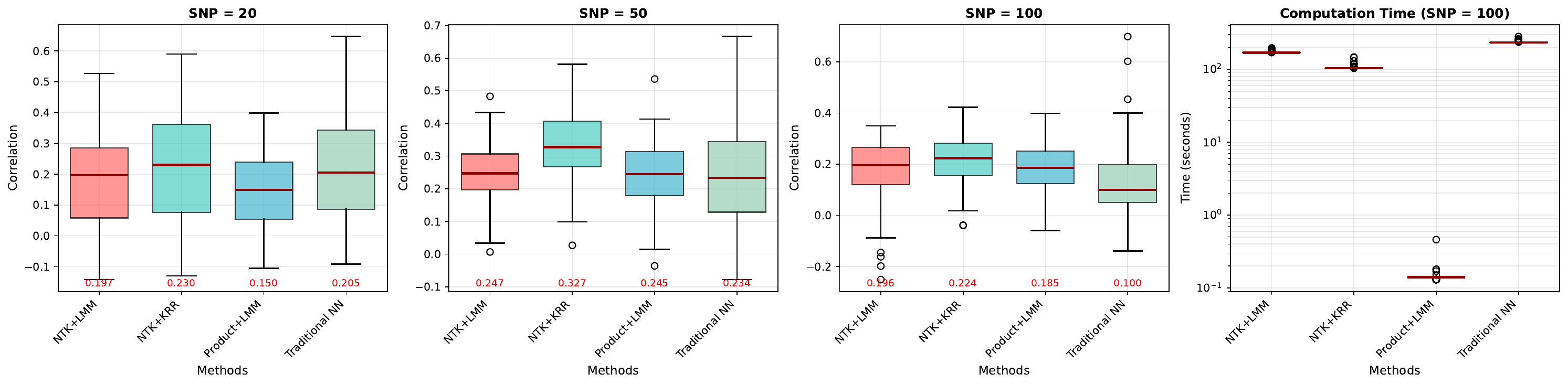}
  \vspace{.2in}
  \caption{Testing correlation and computational time for power model scenario.}
  \label{Simulation2}
\end{figure}

\begin{figure}[H]
  \centering
  \includegraphics[width=\textwidth]{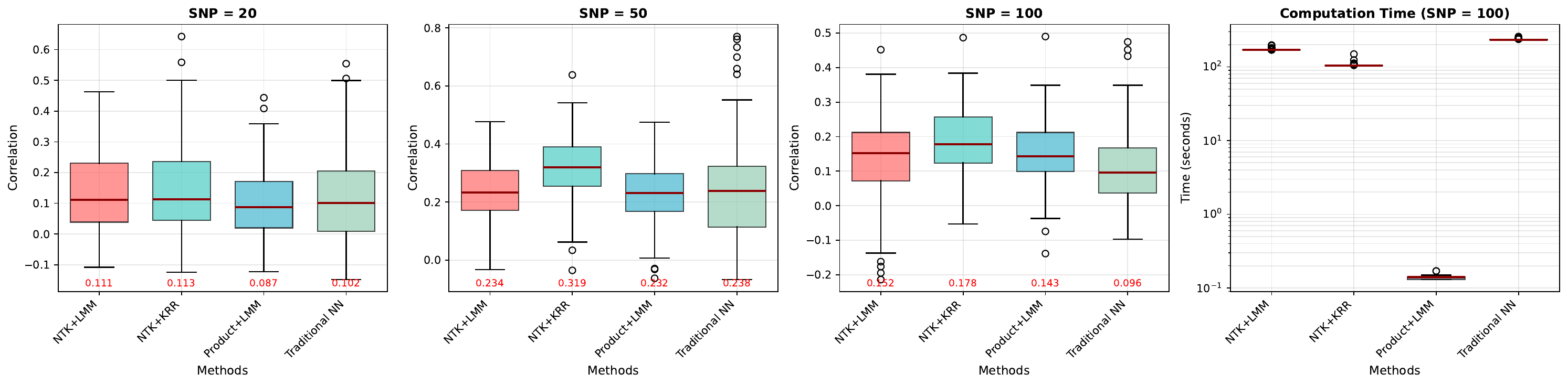}
  \vspace{.2in}
  \caption{Testing correlation and computational time for cosh model scenario.}
  \label{Simulation3}
\end{figure}

\begin{figure}[H]
  \centering
  \includegraphics[width=\textwidth]{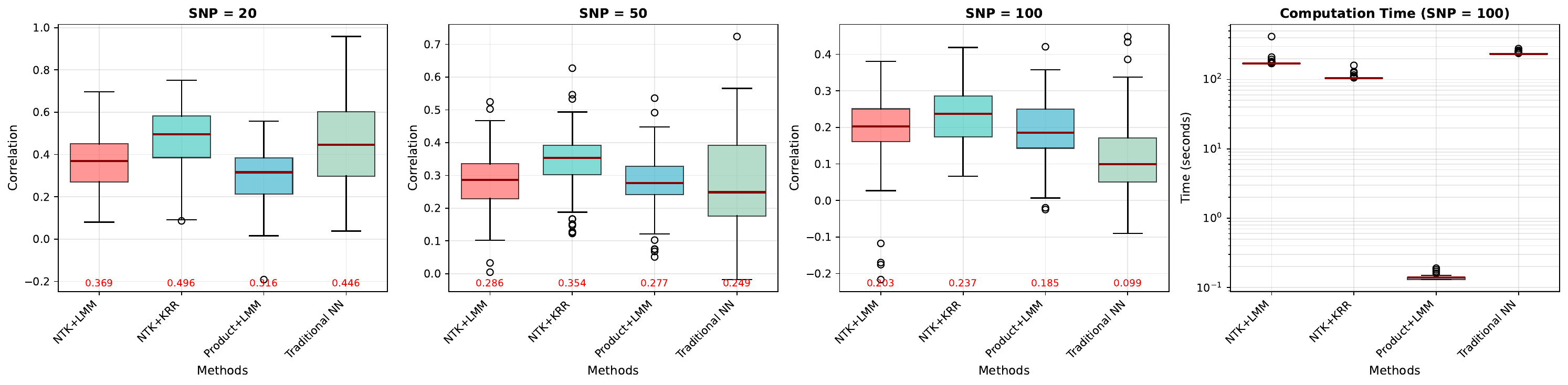}
  \vspace{.2in}
  \caption{Testing correlation and computational time for hyperbolic model scenario.}
  \label{Simulation4}
\end{figure}

\subsection{Pseudocode for Core Algorithms}

\begin{algorithm}[H]
\caption{Empirical NTK Computation (matching \texttt{get\_ntk()} and \texttt{compute\_empirical\_ntk()})}
\begin{algorithmic}[1]
\Require Genotypes $\mathbf{X} \in \mathbb{R}^{n \times p}$, hidden width $m$, depth $L$, chunk size $c$, number of initializations $M$
\Ensure NTK matrix $\mathbf{K} \in \mathbb{R}^{n \times n}$
\State Convert $\mathbf{X}$ to \texttt{torch.float64} tensor
\State Initialize $\mathbf{K} \gets \mathbf{0}_{n \times n}$
\For{$s = 1, \dots, M$}
  \State Initialize network $f_{\boldsymbol{\theta}}$ with architecture: $p \to m \to \cdots \to m \to 1$ ($L$ layers total)
  \State Initialize weights: $\mathbf{W}^{(\ell)} \sim \mathcal{N}(0,1)$ for $\ell=1,\dots,L-1$, $\mathbf{w}^{(L)} \sim \mathcal{N}(0,1)$; biases $\mathbf{0}$
  \For{$i = 1, 1+c, 1+2c, \dots$}
    \State $I \gets \{i,\dots,\min(i+c-1,n)\}$
    \For{$j = 1, 1+c, 1+2c, \dots$}
      \State $J \gets \{j,\dots,\min(j+c-1,n)\}$
      \State Compute Jacobians $\mathbf{J}_1 \in \mathbb{R}^{|I|\times |\boldsymbol{\theta}|}$ for samples in $I$
      \State Compute Jacobians $\mathbf{J}_2 \in \mathbb{R}^{|J|\times |\boldsymbol{\theta}|}$ for samples in $J$
      \State $\mathbf{K}[I,J] \mathrel{+}= \mathbf{J}_1 \mathbf{J}_2^\top$
    \EndFor
  \EndFor
\EndFor
\State $\mathbf{K} \gets \mathbf{K}/M$
\State Symmetrize: $\mathbf{K} \gets (\mathbf{K} + \mathbf{K}^\top)/2$
\State Extract diagonal: $\mathbf{d} \gets \text{diag}(\mathbf{K})^{1/2}$ with $d_i \gets \max(d_i, 10^{-10})$
\State Normalize: $K_{ij} \gets K_{ij}/(d_i d_j)$
\State Scale: $\mathbf{K} \gets \mathbf{K}/p$
\State \Return $\mathbf{K}$
\end{algorithmic}
\end{algorithm}

\begin{algorithm}[H]
\caption{MINQUE-based LMM Training (matching \texttt{knn()})}
\begin{algorithmic}[1]
\Require Phenotypes $\mathbf{y} \in \mathbb{R}^n$, genotype blocks $\{\mathbf{X}^{(r)}\}_{r=1}^R$, kernel types, MINQUE type
\Ensure Variance components $\boldsymbol{\theta}$, fixed effects $\boldsymbol{\beta}$
\For{$r=1,\dots,R$}
  \State Compute kernel $\mathbf{K}^{(r)}$ using specified kernel type
\EndFor
\State Construct variance components $\mathcal{V} = \{\mathbf{I}_n, \mathbf{K}^{(1)}, \dots, \mathbf{K}^{(R)}\}$
\State Initialize $\mathbf{V}$ based on MINQUE type
\State Ensure $\mathbf{V}$ is positive definite using \texttt{near\_pd()}
\State Compute $\mathbf{P} \gets \mathbf{V}^{-1}$
\If{fixed effects present}
  \State Estimate $\boldsymbol{\beta} \gets (\mathbf{F}^\top\mathbf{F})^{-1}\mathbf{F}^\top\mathbf{y}$
  \State Compute residuals $\mathbf{r} \gets \mathbf{y} - \mathbf{F}\boldsymbol{\beta}$
\Else
  \State $\boldsymbol{\beta} \gets \mathbf{0}$, $\mathbf{r} \gets \mathbf{y}$
\EndIf
\State Form system: $A_{ij} \gets \text{tr}(\mathbf{P}\mathbf{V}_i\mathbf{P}\mathbf{V}_j)$, $b_i \gets \mathbf{r}^\top\mathbf{P}\mathbf{V}_i\mathbf{P}\mathbf{r}$
\State Solve $\boldsymbol{\theta} \gets \mathbf{A}^{-1}\mathbf{b}$
\If{constrained}
  \State $\theta_i \gets \max(0, \theta_i)$ for all $i$
\EndIf
\State \Return $\boldsymbol{\theta}$, $\boldsymbol{\beta}$
\end{algorithmic}
\end{algorithm}

\begin{algorithm}[H]
\caption{LMM Prediction (matching \texttt{knn\_predict()})}
\begin{algorithmic}[1]
\Require Training $\mathbf{y}_{\text{train}}$, variance components $\boldsymbol{\theta}$, kernels for train/test
\Ensure Predictions $\hat{\mathbf{y}}_{\text{test}}$
\State Construct combined kernel: $\mathbf{K}_{\text{RE}} \gets \sum_{k \geq 1} \theta_k \mathbf{V}_k$
\State Extract blocks: $\mathbf{K}_{\text{train}}$, $\mathbf{K}_{\text{test,train}}$
\State Form covariance: $\mathbf{C} \gets \mathbf{K}_{\text{train}} + \theta_0 \mathbf{I}$
\If{fixed effects present}
  \State Compute residuals: $\mathbf{r} \gets \mathbf{y}_{\text{train}} - \mathbf{F}_{\text{train}}\boldsymbol{\beta}$
\Else
  \State $\mathbf{r} \gets \mathbf{y}_{\text{train}}$
\EndIf
\State Solve: $\boldsymbol{\alpha} \gets \mathbf{C}^{-1} \mathbf{r}$
\State Predict: $\hat{\mathbf{y}}_{\text{test}} \gets \mathbf{F}_{\text{test}}\boldsymbol{\beta} + \mathbf{K}_{\text{test,train}} \boldsymbol{\alpha}$
\State \Return $\hat{\mathbf{y}}_{\text{test}}$
\end{algorithmic}
\end{algorithm}

\begin{algorithm}[H]
\caption{Kernel Ridge Regression with NTK (matching \texttt{krr\_with\_ntk()})}
\begin{algorithmic}[1]
\Require Training data $(\mathbf{X}_{\text{train}}, \mathbf{y}_{\text{train}})$, test data $\mathbf{X}_{\text{test}}$, NTK parameters $(m, L, M)$
\Ensure Predictions $\hat{\mathbf{y}}_{\text{test}}$, optimal $\lambda^*$
\State Compute NTK matrices using Algorithm 1:
\State \quad $\mathbf{K}_{\text{train}} \gets \text{NTK}(\mathbf{X}_{\text{train}}, \mathbf{X}_{\text{train}})$
\State \quad $\mathbf{K}_{\text{test,train}} \gets \text{NTK}(\mathbf{X}_{\text{test}}, \mathbf{X}_{\text{train}})$
\State \textbf{Cross-validation for $\lambda$:}
\For{each $\lambda \in \{0.001, 0.01, 0.1, 1.0, 10.0\}$}
  \For{each fold $k = 1, \dots, 5$}
    \State Split training data into train/validation for fold $k$
    \State Compute: $\boldsymbol{\alpha}_k \gets (\mathbf{K}_{\text{fold-train}} + \lambda \mathbf{I})^{-1} \mathbf{y}_{\text{fold-train}}$
    \State Predict: $\hat{\mathbf{y}}_{\text{fold-val}} \gets \mathbf{K}_{\text{fold-val,train}} \boldsymbol{\alpha}_k$
    \State Record MSE for fold $k$
  \EndFor
  \State Compute mean MSE across folds for $\lambda$
\EndFor
\State Select $\lambda^* \gets \arg\min_\lambda \text{mean MSE}(\lambda)$
\State \textbf{Final prediction:}
\State Solve using Cholesky decomposition:
\State \quad $\mathbf{L}\mathbf{L}^\top = \mathbf{K}_{\text{train}} + \lambda^* \mathbf{I}$
\State \quad Solve $\mathbf{L}\mathbf{z} = \mathbf{y}_{\text{train}}$ for $\mathbf{z}$
\State \quad Solve $\mathbf{L}^\top\boldsymbol{\alpha} = \mathbf{z}$ for $\boldsymbol{\alpha}$
\State Predict: $\hat{\mathbf{y}}_{\text{test}} \gets \mathbf{K}_{\text{test,train}} \boldsymbol{\alpha}$
\State \Return $\hat{\mathbf{y}}_{\text{test}}$, $\lambda^*$
\end{algorithmic}
\end{algorithm}

\end{document}